\documentclass[11pt]{CGC2}

\usepackage[english]{babel}
\usepackage{verbatim}
\usepackage{graphicx}
\usepackage{color}
\usepackage{longtable}

\newcommand{\BAA}{Bound $\A$ }
\newcommand{\Aq}{A_q(n,d)}

\newcommand{\A}{{\mathcal A}}
\newcommand{\B}{{\mathcal B}}

\newcommand{\w}{\mathrm{w}}

\usepackage[english]{babel}
\usepackage{amsmath}
\usepackage{amssymb}
\usepackage{amsthm}
\usepackage[mathcal]{eucal}
\usepackage[latin1]{inputenc}

\begin{document}

\Logo{}

\begin{frontmatter}

\title{Some bounds on the size of codes}

{\author{Emanuele Bellini}} {\tt{(eemanuele.bellini@gmail.com)}}\\
{Department of Mathematics, University of Trento, Italy.}

{\author{Eleonora Guerrini}}
{{\tt (guerrini@lirmm.fr)}}\\
{{LIRMM, Université de Montpellier 2, France.}}

{\author{Massimiliano Sala}} {\tt{(maxsalacodes@gmail.com)}}\\
{Department of Mathematics, University of Trento, Italy.}

\runauthor{E.~Bellini, E.~Guerrini, M.~Sala}

\begin{abstract}
We present some upper bounds on the size of non-linear codes and their restriction to systematic codes and linear codes. These bounds are independent of other known theoretical bounds, e.g. the Griesmer bound, the Johnson bound or
the Plotkin bound, and one of these is actually an improvement of a bound by Litsyn and Laihonen. Our experiments show that in some cases (the majority of cases for some $q$) our bounds provide the best value, compared to all other theoretical bounds.
\end{abstract}

\end{frontmatter}

\section{Introduction}

The problem of bounding the size of a code depends heavily on the code family that we are considering. In this paper we are interested in three types of codes: linear codes, systematic codes and non-linear codes. Referring to the subsequent section for rigorous definitions, with {\bf linear codes} we mean linear subspaces of $(\FF_q)^n$, while with {\bf non-linear codes} we mean (following consolidated tradition) codes that are not necessarily linear. In this sense, a linear code is always a non-linear code, while a non-linear code may be a linear code, although it is unlikely. Systematic codes form a less-studied family of codes, whose definition is given in the next section. 
Modulo code equivalence all (non-zero) linear codes are systematic and all systematic codes are non-linear. In some sense, systematic codes stand in the middle between linear codes and non-linear codes. The size of a systematic code is directly comparable with that of a linear code, since it is a power of the size of $\FF_q$.

In this paper we are interested only in {\bf theoretical bounds}, that is, bounds on the size of a code that can be obtained by a closed-formula expression, although other algorithmic bounds exist (e.g. the Linear Programming bound \cite{CGC-cd-art-Dels73}).
The algebraic structure of linear codes would suggest the knowledge of a high number of bounds strictly for linear codes, and only a few bounds for the other case. Rather surprisingly, the academic literature reports only one bound for linear codes, the Griesmer bound (\cite{CGC-cd-art-griesm60}), no bounds for systematic codes and many bounds for  non-linear codes. Among those, we recall: the Johnson bound (\cite{CGC-cd-art-john62},\cite{CGC-cd-art-john71},\cite{CGC-cd-book-huffmanPless03}), the Elias-Bassalygo bound (\cite{CGC-cd-art-bass65},\cite{CGC-cd-book-huffmanPless03}), the Levenshtein bound (\cite{CGC-cd-art-lev98}), the Hamming (Sphere Packing) bound and the Singleton bound (\cite{CGC-cd-book-pless98}), and the Plotkin bound (\cite{CGC-cd-art-plotk60}, \cite{CGC-cd-book-huffmanPless03}).\\
Since the Griesmer bound is specialized for linear codes, we would expect it to beat the other bounds,
but even this does not happen, except in some cases. So we have an unexpected situation where the bounds holding for the more general case are numerous and beat bounds holding for the specialized case.\\
\indent
In this paper we present one (closed-formula) bound (\BAA) for a large parte of non-linear codes (including all systematic codes), which is an improvement of a bound by Litsyn and Laihonen in \cite{CGC-cod-art-litlai98}. The crux of our improvement is a preliminary
result presented in Section \ref{secMainThm}, while in Section 4 we are able to prove \BAA.
Then we restrict \BAA\ to the systematic/linear case and compare it with all the before-mentioned bounds by computing their values for a large set of parameters (corresponding to about one week of computations with our computers). Our findings are in favour of \BAA and are reported in Section 5. For large values of $q$, our bound provides the best value in the majority of cases.\\
The only bound that we never beat is Plotkin's, but its range is very small (the distance has to be at least $d > n(1-1/q)$) and the cases falling in this range are a tiny portion with large $q$'s.\\
\indent
For standard definitions and known bounds, the reader is directed to the original articles or to any recent good book, e.g. \cite{CGC-cd-book-huffmanPless03} or \cite{CGC-cd-book-pless98}.


\section{Preliminaries}
\label{prel}
We first recall a few definitions.\\
Let $\FF_q$ be the finite field with $q$ elements, where $q$ is any power of any prime.\\                                                                                                                                        
Let $n\geq k\geq 1$ be integers.
Let $C \subseteq \FF_q^n, C \ne \emptyset$. We say that $C$ is an $(n,q)$ {\bf code}. Any $ c\in C$  is  a  {\bf word}. Note that here and afterwards a ``code'' denotes what is called a ``non-linear code'' in the introduction. \\
Let $\phi:(\mathbb{F}_q)^k \rightarrow (\mathbb{F}_q)^n$ be an injective function and let $C = {\rm Im}(\phi)$. We say that $C$ is an $(n,k,q)$ {\bf systematic code} if $\phi(v)_i = v_i$ for any $v \in (\FF_q)^k$ and any 
$1 \le i \le k$.
If $C$ is a vector subspace of $(\mathbb{F}_q)^n$, then $C$ is a  {\bf linear} code. Clearly any non-zero linear code is equivalent to a systematic code. \\
From now on, $\FF$ will denote $\FF_q$ and $q$ is understood.\\
We denote with $\d(c,c')$ the {\bf (Hamming) distance} of two words $c,c' \in C$, which is the number of different components between $c$ and $c'$. We denote with $d$ a number such that $1 \le d \le n$ to indicate the {\bf distance of a code}, which is $d = \min_{c,c' \in C,c \ne c'}\{\d(c,c')\}$. Note that a code with only one word has, by convention, distance equal to infinity. The whole $\FF^n$ has distance $1$, and $d = n$ in a systematic code is possible only if $k=1$.\\
From now on, $n,k$ are understood.
\begin{definition}\label{ball}
Let $l,\,m \in \NN$ such that $l \leq m$. In $\mathbb{F}^m$, we denote  by $B_x(l,m)$ the set of vectors  with distance from the word $x$ less than or equal to $l$, and we call it the \textbf{ball} 
centered in $x$ of radius $l$.\\
For conciseness, $B(l,m)$ denotes the ball centered in the zero vector.
\end{definition}
\noindent
Obviously, $B(l,m)$ is the set of vectors of weight less than or equal to $l$ and
$$
|B(l,m)|  \; = \; \sum_{j=0}^{l} \binom{m}{j}(q-1)^j.
$$
We also note that any two balls having the same radius over the same field contain the same number of vectors.
\begin{definition}
 The number $A_q (n, d)$ denotes the maximum number of words in a code over $\FF_q$ of length $n$ and distance $d$.  
\end{definition}


\section{A first result for a special code family}\label{secMainThm}

The maximum number of words in an $(n,d)$ code can be smaller than $A_q (n, d)$ if we have extra constraints on the weight of words.
The following result is an example and it will be instrumental of the proof of \BAA.
\begin{theorem}[]
  \label{thmEPS}
  Let $C$ be an $(n,d)$-code over $\FF^n$. Let  $\epsilon \ge 1$ be such that for any $c \in C$ we have $\w(c) \ge d+\epsilon$.
  Then
  $$ |C| \le A_q(n,d) - \frac{|B(\epsilon,n)|}{|B(d-1,n)|}$$
\end{theorem}

\begin{proof}
$C$ belongs to the set of all codes with distance $d$ and contained in $\FF^{n} \setminus B_0(d+\epsilon-1,n)$.
Let $D$ be any code of the largest size in this set, then
  \begin{align}\label{eq1}
    |C| \le |D|   
  \end{align}
Clearly, any word $c$ of $D$ has weight $\w(c)\geq d+\epsilon$.
  Consider also $\bar{D}$, the largest code over $\FF^{n}$ of distance $d$ such that $D \subseteq \bar{D}$.
By definition, the only words of $\bar{D}$ of weight greater than $d+\epsilon-1$ are those of $D$, while all other words of $\bar{D}$ are confined to the ball $B_0(d+\epsilon-1,n)$.
Thus
  \begin{align}\label{eq0}
    |C| \le |D| \le |\bar{D}| \le A_q(n,d)
  \end{align}
  and
  $$\bar{D} \setminus D \subseteq B_0(d+\epsilon-1,n)$$
  Let $\rho = d-1$ and $r = d+\epsilon-1$, so that $r-\rho = \epsilon$, and let $N = \bar{D} \cap B_0(r,n)$. We have: 
  \begin{align}\label{eq2}  
    D = \bar{D} \setminus N, \qquad |D| = |\bar{D}| - |N|
  \end{align}

  We are searching for a lower bound on $|N|$, in order to have an upper bound on $|D|$.
  We start with proving
  \begin{align}\label{eqBx}
    B_0(r-\rho,n) \subseteq \bigcup_{x \in N}B_{x}(\rho,n)
  \end{align}
  Consider $y \in B_0(r-\rho,n)$. If for all $x \in N$ we have that $y \notin B_x(\rho,n)$, then $y$ is a vector whose distance from $N$ is at least $\rho+1$.
  Since $y \in B_0(r-\rho,n)$, also its distance from $\bar{D}\setminus N$ is at least $\rho+1$. Therefore, the distance of $y$ from the whole $\bar{D}$ is at least
  $\rho+1=d$ and so we can obtain a new code $\bar{D} \cup \{y\}$ containing  $D$ and with distance $d$, contradicting
  the fact that $|\bar{D}|$ is the largest size for such a code in $\FF^{n}$. So, (\ref{eqBx}) must hold.

A direct consequence of  (\ref{eqBx}) is
  \begin{align*}
    & |N|\cdot|B_{x}(\rho,n)| \ge |B_0(r-\rho,n)| \,,
  \end{align*}
  which gives
  \begin{align}\label{eq3}
    |N| \ge \frac{|B_0(r-\rho,n)|}{|B_{x}(\rho,n)|} = \frac{|B_0(\epsilon,n)|}{|B_{x}(d-1,n)|} 
  \end{align}
  Using (\ref{eq1}), (\ref{eq0}), (\ref{eq2}) and (\ref{eq3}), we obtain the desired bound:
  \begin{align*}
    |C| \le |D| &= |\bar{D}| - |\bar{D} \cap B_0(d+\epsilon-1,n)| \\
                &  \le A_q(n,d) - \frac{|B_0(\epsilon,n)|}{|B_{x}(d-1,n)|}
  \end{align*}
\end{proof}


\section{An improvement of the Litsyn-Laihonen bound}

In 1998 Litsyn and Laihonen prove a bound for non-linear codes: \\ Theorem~1 of \cite{CGC-cod-art-litlai98}, which we write with our notation as follows.
\begin{theorem}[Litsyn-Laihonen bound]
 \label{thmLLboundold}
 Let $1 \le d \le n$. Let $t \in \mathbb{N}$ be such that $t \le n-d$. Let $r \in \NN$ be such that $d-2r \le n-t$, $0 \le r \le t$ and $0 \le r \le \frac{1}{2}d$. Then
 $$A_q(n,d) \le \frac{q^t}{|B(r,t)|}A_q(n-t,d-2r)$$
\end{theorem}

Let $C$ be an $(n,d)$-code over $\FF$, let $k=\lfloor \log_q(|C|) \rfloor$. We say that $C$ is \emph{systematic-embedding} if $C$ contains a systematic code $D$ with size $|D|=q^k$. Obviously a systematic code is systematic-embedding with $D=C$. Moreover if the code is linear then $k$ is the dimension of $C$.\\
All known families of maximal codes are either systematic codes or systematic-embedding codes (see e.g., \cite{CGC-cd-art-preparata}, \cite{CGC-kerd72} and \cite{CGC-cd-art-goethal}).
\\
We are ready to show a strengthening of Theorem \ref{thmLLboundold} restricted to systematic-embedding codes: \BAA. In the proof we follow initially the outline of the proof of \cite{CGC-cod-art-litlai98}[Theorem 1] and then we apply Theorem \ref{thmEPS}.
\begin{theorem}[Bound $\A$]
 \label{thmLLbound}
 Let $1 \le d \le n$. Let $t \in \mathbb{N}$ be such that $t \le n-d$. Let $r \in \NN$ be such that $d-2r \le n-t$, $0 \le r \le t$ and $0 \le r \le \frac{1}{2}d$. Suppose that there is an $(n,d)$-code $C$ over $\FF$ such that $|C|=A_q(n,d)$ and $C$ is systematic-embedding. Let $t \le k = \lfloor \log_q(|C|) \rfloor $. Then
 $$
   A_q(n,d) \le \frac{q^t}{|B(r,t)|} \left(A_q(n-t,d-2r) - \frac{|B(r,n-t)|}{|B(d-2r-1,n-t)|} + 1 \right)
 $$
\end{theorem}
\begin{proof}
 We consider an $(n,d)$ code $C$ such that $|C| = A_q(n,d)$ and $C$ is systematic-embedding. By hypothesis $C$ must exist.
We number all words in $C$ in any order: $C=\{ c_i \mid 1\leq\ i\leq A_q(n,d)\}$. \\ We indicate the $i$-th word with $c_i = (c_{i,1},\dots,c_{i,n})$.
 We puncture $C$ as follows:
 \begin{enumerate}
  \item[(i)] we choose any $t$ columns among the $k$ columns of the systematic part of $C$, $1\leq j_1,\dots,j_t\leq n$; since two codes are equivalent w.r.t. column permutations we suppose $j_1=1,\dots,j_t=t$.\\
             Let us split each word $c_i \in C$ in two parts
            \begin{align*}
                \tilde{c_i} = (c_{i,1},\dots,c_{i,t}) \quad \bar{c_i}   = (c_{i,t+1},\dots,c_{i,n}),\qquad \mbox{ so} \quad  c_i = (\tilde{c_i},\bar{c_i}).
            \end{align*}
  \item[(ii)] We choose a $z\in \FF^t$.
  \item[(iii)] We collect in $I$ all $i$'s s.t.  $d(z,\tilde{c_i}) \le r$;
  \item[(iv)] We delete the first $t$ components of $\{c_i \mid i \in I\}$.
 \end{enumerate}
Then the punctured code $\bar{C}_z$ obtained by (i),(ii),(iii) and (iv) is:
 \begin{align*}
   \bar{C}_z = \{\bar{c_i} \mid i\in I\}=\{\bar{c_i} \mid 1\leq i\leq A_q(n,d), d(z,\tilde{c_i}) \le r\}
 \end{align*}
 We claim that we can choose $z$ in such a way that $\bar{C}_z$ is equivalent to a code with the following properties:
 \begin{align}
  \label{eqLeng} & \bar{n} = n-t \\ 
  \label{eqDist} & \bar{d} \ge d-2r \\
  \label{eqCard} & |\bar{C}_z| \ge \frac{|C|}{q^t}|B(r,t)| \\
  \label{eqWeig} & \w(\bar{c_i}) \ge d-r \text{ for all } \bar{c_i} \ne 0 
 \end{align}
 (\ref{eqLeng}) is obvious.
 As regards (\ref{eqDist}), note that  $\d(c_i,c_j) = \d(\tilde{c_i},\tilde{c_j}) + \d(\bar{c_i},\bar{c_j}) \ge d$ and also that  $\tilde{c_i},\tilde{c_j}\in B_z(r,t)$  implies $\d(\tilde{c_i},\tilde{c_j}) \le 2r$.
 Therefore for any $i\ne j$
 $$
   2r + \d(\bar{c_i},\bar{c_j}) \ge \d(\tilde{c_i},\tilde{c_j}) + \d(\bar{c_i},\bar{c_j}) \ge d \,.
 $$
The proof of  (\ref{eqCard}) is more involved and we need to consider the average number $M$ of the $i$'s such that $\tilde{c_i}$ happens to be in a sphere of radius $r$ (in $\FF^{t}$). 
The average is taken over all sphere centers, that is, all vectors $x$'s in $\FF^t$, so that 
$$
  M= \frac{1}{|\FF^t|} \sum_{x\in \FF^t} |\{i \mid 1\leq i\leq A_q(n,d), \tilde{c_i} \in B_x(r,t)\}| \,.
$$
Let us define a function:
%
%
 $$
   \psi: \FF^t \times \FF^t \longrightarrow \{0,1\}, \qquad
   \psi(x,y) =
   \bigg \{
   \begin{array}{rl}
   1, & \d(x,y) \le r \\
   0, & \text{otherwise} \\
   \end{array}.
 $$
 Then we can write $M$ and $|B_y(r,t)|$ (for any $y\in \FF^t$) as
 $$
   M  = \frac{1}{q^t} \sum_{x \in \FF^{t}} \sum_{i=1}^{\Aq}    \psi(x,\tilde{c_i}) \qquad
   |B_y(r,t)|= \sum_{x  \in \FF^t} \psi(x,y) \,.
 $$
 By swapping variables we get
 \begin{align*}
   M  = \frac{1}{q^t} \sum_{x \in \FF^{t}} \sum_{i=1}^{\Aq}    \psi(x,\tilde{c_i})
     = \frac{1}{q^t} \sum_{i=1}^{\Aq}     \sum_{x  \in \FF^t} \psi(x,\tilde{c_i})
     = \frac{\Aq}{q^t} |B_{\tilde{c_i}}(r,t)|    \,.
 \end{align*}
 This means that there exists $\hat{x} \in \FF^t$ such that
 $$
     |\{i \mid 1\leq i\leq A_q(n,d), \tilde{c_i} \in B_{\hat x}(r,t)\}| \geq M \geq \frac{\Aq}{q^t} |B(r,t)| \,. 
 $$
In other words, there are at least  $\frac{|C|}{q^t} |B(r,t)|$ $c_i$'s such that their $\tilde{c_i}$'s are contained in  $B_{\hat x}(r,t)$.
Distinct $c_i$'s may well give rise to the same $\tilde{c_i}$'s, but they always correspond to distinct $\bar{c_i}$'s (see the proof of (\ref{eqDist})),
so there are at least $\frac{|C|}{q^t} |B(r,t)|$ (distinct) $\bar{c_i}$'s such that their corresponding $\tilde{c_i}$'s
fall in $B_{\hat x}(r,t)$. By choosing $z=\hat{x}$ we then have at least   $\frac{|C|}{q^t} |B(r,t)|$ (distinct) codewords of $\bar{C}_z$ and so (\ref{eqCard}) follows.

We claim that (\ref{eqWeig}) holds if $0\in C$ and $z = 0$. Infact:
\begin{align*}
  & \w(c) = \d(0,c) \geq d, \,\quad \forall c \in C \mbox{ such that } c\neq 0. \\
  & z = 0 \quad \implies \qquad y \in B_z(r,t) \iff \w(y) \le r.
\end{align*}
  As a consequence, any nonzero word $c_i = (\tilde{c_i},\bar{c_i})$ of weight at most $r$ in $\tilde{c_i}$ has weight at least $d-r$ in the other $n-t$ components.\\
 If $0 \notin C$ or $z \ne 0$ we consider a code $C+v$ equivalent to $C$, by choosing the translation $v$ in the following way. By hypothesis of systematic-embedding there exists $\hat{c} \in C$ such that its first $t$ coordinates form the vector $\hat{x}$. By considering $v = \hat{c}$ we obtain the desired code, thus (\ref{eqWeig}) is proved.\\
 
 Now we call $X$ the largest $(\bar{n},d-2r)$-code containing the zero word and such that $\w(\bar{x}) \ge d-r = (d-2r) +r$, $\forall \bar{x} \in X$. Observe that $X$ satisfies (\ref{eqLeng}), (\ref{eqDist}), (\ref{eqCard}), (\ref{eqWeig}) and so $|X| \ge |\bar{C_z}|$. Then we can apply Theorem \ref{thmEPS} to $X \setminus \{0\}$ and $\epsilon = r$, and obtain the following chain of inequalities:
 \begin{align*}
  \frac{|C|}{q^t} |B(r,t)| \le |\bar{C_z}| \le |X| \le A_q(\bar{n},d-2r) - \frac{|B(r,\bar{n})|}{|B(d-2r-1,\bar{n})|} + 1
 \end{align*}
 and since $|C| = A_q(n,d)$ we have the bound:
 \begin{align*}
  A_q(n,d) \le \frac{q^t}{|B(r,t)|}\left(A_q(\bar{n},d-2r) - \frac{|B(r,\bar{n})|}{|B(d-2r-1,\bar{n})|} + 1\right).
 \end{align*}
\end{proof}


\subsection{Systematic case}

When we restrict ourselves into the systematic/linear case, then the value $A_q(n,d)$ can only be a power of $q$, and if the dimension of the code $C$ is $k$, then $A_q(n,d) = q^k$. By choosing $t=k$ we have the following corollary:


\begin{corollary}[Bound $\B$]
  \label{boundB}
  Let $k,d,r \in \NN, d \ge 2,k \ge 1$. Let $n$ be such that there exists an  $(n,k,q)$ systematic code $C$ with distance at least $d$.\\
  If $0 \le r \le \min\{ \lfloor \frac{d-1}{2} \rfloor,k\} $, then
  $$ |B(r,k)| \le A_q(n-k,d-2r) - \frac{|B(r,n-k)|}{|B(d-2r-1,n-k)|} + 1.$$
\end{corollary}
In the systematic/linear case the Litsyn-Laihonen bound becomes:
$$|B(r,k)| \le A_q(n-k,d-2r).$$
Easy computations can be done in the case $d=3$, since in this case $r$ can be at most $1$, so that:
\begin{itemize}
 \item $|B(1,k)| = (q-1)k + 1$
 \item $A_q(n-k.d-2r) = A_q(n-k,1) = q^{n-k}$
 \item $|B(1,n - k)| = (q - 1)(n - k) + 1$
 \item $|B(d - 2r - 1, n - k)| = |B(0, n - k)| = 1$
\end{itemize}
Our bound then reduces to:
$$0 \le q^{n-k} - (q - 1)n - 1$$
which is stronger then the Litsyn-Laihonen bound, which in the case $d=3$ reduces to:
$$0 \le q^{n-k} - (q - 1)k - 1.$$

\section{Experimental comparisons with other upper bounds, remarks and conclusion}

We have analyzed the case of linear codes, implementing Bound $\B$. The algorithm to compute the bound takes as inputs $n,d$, and returns the largest $k$ (checks are done until $k=n-d+1$) such that the inequality of the bound holds. If the inequality always holds in this range, $n-d+1$ is returned. Then we compared our upper bound on $k$ with other bounds, restricting those which hold in the general non-linear case to the systematic case. In particular they give a bound on $A_q(n,d)$ instead of a bound on $k$. As a consequence, for example, if the Johnson bound returns the value $A_q(n,d)$ for a certain pair $(n,d)$, then we compare our bound with the value $\lfloor \log_q(A_q(n,d)) \rfloor$, which is the largest power $s$ of $q$ such that $q^s \le A_q(n,d)$.\\
The inequality in Theorem \ref{boundB} involves the value $A_q(n-k,d-2r)$, which is the maximum number of words that we can have in a \emph{non-linear} code of length $n-k$ and distance $d-2r$. To implement Bound $\B$ it is necessary to compute $A_q(n-k,d-2r)$; when this value is unknown (we use known values only in the binary case for $n = 3,\dots,28 , d = 3,\dots,16$), we return instead an upper bound on it, choosing the best between the Hamming (Sphere Packing), Singleton, Johnson, and Elias bound (the Plotkin bound is used when possible). Even though it is a very strong bound, we do not use the Levenshtein bound because it is very slow as $n$ grows. This means that if better values of $A_q(n-k,d-2r)$ can be found, then Bound $\B$ could return even tighter results.\\

Table \ref{tabStat1} and \ref{tabStat2} show a comparison between all bounds' performances, except for Plotkin's, due to its restricted range. For each bound and for each $q=2,3,4,5,7,8,9,11,13,16,17,19,23,25,27,29$ we have computed, in the range $n=3,\dots,100$ and $d=3,\dots,n-1$, the percentage of cases the bound is the ``best'' known bound between Bound $\B$, the Griesmer, Johnson, Levenshtein, Elias, Hamming and Singleton bound. Both wins and draws are counted in the percentage, since more than one bound may reach the best known bound, and in this case we increased the percentage of each best bound. \\
For each $q$ the most performing bound is in bold. Up to $q=7$ the Levenshtein bound is the most performing. From $9 \le q \le 29$ we have that Bound $\B$ is the most performing bound, and in particular, in the case $q=29$, it is the best known bound almost $91\%$ of the times. \\ 

Table \ref{tabBnewbounds}, instead, shows some cases (one per each $q=7,\dots,29$) where Bound $\B$ beats all other known bounds. This happens from $q=7$, for the range of $n$ considered. The letters B, J, H, G, E, S and L stands respectively for Bound $\B$, Johnson, Hamming (Sphere Packing), Griesmer, Elias, Singleton, and Levenshtein bound. It can be seen that there are some cases where Bound $\B$ is tight, as for the parameters $(9,17,7)$, for which there exist a code with distance $10$.\\

Tables \ref{tabB3-100}, and \ref{tabB3-100bis} give emphasis to the number of times Bound $\B$ improves the best known bound (thus the cases where it beats all other bounds). In the considered range Bound $\B$ starts to beat all other bounds from $q=7$. \\
The third row of Tables \ref{tabB3-100} and \ref{tabB3-100bis} shows how many times (percentage over the number of draws and wins) the value $\delta = \frac{|B(r,n-k)|}{|B(d-2r-1,n-k)|}$ is different from zero. Informally, we can view $\delta$ as the probability to randomly pick up a word of weight less than $r$ from a ball of radius $d-2r-1$. We can notice that this percentage is very high, which means that a weaker version of Bound $\B$, which is similar to the Litsyn-Laihonen bound for systematic codes, could be used, by simply searching the largest $k$ satisfying:
$$ |B(r,k)| \le A_q(n-k,d-2r) + 1$$
It is curious to notice that in all the wins we have $\delta=0$, and that $\delta=0$ also $38094$ times over the $46967$ ties and wins. This means that the weaker version of Bound $\B$ is sufficient to obtain most of the wins and ties in the investigated cases.\\
We note that in general, if $r$ is greater than $d$, we expect $\delta$ big, decreasing very quickly as $r$ increases, holding $d$ fixed; this happens since $|B(x,k)|$ decreases following a gaussian distribution (roughly approximating $|B(x,k)|$ with a factorial), and so any time we subtract $2r$ the decrease is doubled.\\
The fourth row of Tables \ref{tabB3-100} and \ref{tabB3-100bis} shows the ratio between the number of times the Plotkin bound has been used to bound $A_q(n-k,d-2r)$ and the number of draws and wins. Third and fourth row show values which are close for $q$ small and gets further as $q$ grows. This happens because almost all the times that the weaker version (with $\delta = 0$) of Bound $\B$ ties with the best known bound, a strong bound on $A_q(n-k,d-2r)$ must be used, and the strongest bound is Plotkin's, which though has a smaller range of applicability as $q$ grows.\\
We report in the fifth row of Tables \ref{tabB3-100} and \ref{tabB3-100bis} the fact that the maximum ratio $d/n$ reached in the wins of Bound $\B$ grows up to the value $0.64$ and then seems to get stabilized toward $0.5$. This means that Bound $\B$ is a very strong bound for distances which are no more than $\frac{2}{3}$ of the length $n$ for small values of $q$, and no more than half of the length $n$ for bigger values of $q$.\\

Comparisons have been made using inner MAGMA (\cite{CGC-MAGMA}) implementations of known upper bounds, except for the Johnson bound. For this bound we noted that the inner MAGMA implementation could be improved and so we used our own MAGMA implementation for this bound. 
\section*{Acknowledgements}
The first two authors would like to thank the third author (their supervisor). Partial results appear in \cite{CGC-cd-phdthesis-ele}. The authors would like to thank: Ludo Tolhuizen (Philips Group Innovation, Research), the MAGMA group and in particular John Cannon.



\bibliography{RefsCGC}


\newpage

The following tables show the results computed in the range $n = 3,\dots,100$, $d=3,\dots,n-1$.\\
\begin{table}[h]
\begin{center}
\begin{tabular}{p{0.18\textwidth}|llllllll}
\hline
$q$ & 2 & 3 & 4 & 5 & 7 & 8 & 9 & 11\\
\hline
Bound $\B$  & 38.02 & 31.20 & 31.20 & 31.94 &  40.73 & 48.64 & \bf{55.27} &  \bf{66.44} \\
Johnson     & 40.65 & 31.18 & 33.50 & 35.13 &  35.70 & 35.51 & 35.09 &  33.26 \\
Hamming     & 18.12 & 15.65 & 16.37 & 16.35 &  16.03 & 15.88 & 15.57 &  14.69 \\
Griesmer    & 56.32 & 39.83 & 32.32 & 29.14 &  30.91 & 36.97 & 43.28 &  55.15 \\
Levenshtein & \bf{72.65} & \bf{69.68} & \bf{66.27} & \bf{64.02} &  \bf{60.80} & \bf{58.24} & 54.47 &  46.26 \\
Elias       & 6.859 & 32.28 & 38.27 & 40.02 &  40.82 & 40.14 & 37.24 &  31.37 \\
Singleton   & 0.000 & 0.021 & 0.084 & 0.189 &  0.610 & 0.926 & 1.241 &  3.619 \\
\hline
\end{tabular}
\end{center}
\caption{When each bound is the best for $2 \le q \le 11$.}
\label{tabStat1}
\end{table}

\begin{table}
\begin{center}
\begin{tabular}{p{0.18\textwidth}|llllllll}
\hline
$q$ & 13 & 16 & 17 & 19 & 23 & 25 & 27 & 29\\
\hline
Bound $\B$  &  \bf{76.43} &   \bf{81.61} & \bf{82.75} &  \bf{85.42} &    \bf{88.11} &  \bf{88.72} &  \bf{89.40} &  \bf{90.77} \\
Johnson     &  30.80 &   26.61 & 24.87 &  21.88 &    17.08 &  15.51 &  14.37 &  13.34 \\
Hamming     &  13.59 &   11.91 & 11.26 &  10.12 &    8.269 &  7.553 &  7.048 &  6.606 \\
Griesmer    &  63.39 &   71.91 & 72.27 &  71.94 &    69.79 &  69.43 &  68.65 &  67.87 \\
Levenshtein &  39.93 &   32.86 & 30.65 &  27.50 &    22.62 &  20.70 &  19.44 &  18.37 \\
Elias       &  27.06 &   21.84 & 20.01 &  17.59 &    12.48 &  10.84 &  9.657 &  8.689 \\
Singleton   &  4.439 &   4.629 & 6.985 &  6.712 &    10.08 &  12.01 &  14.12 &  18.01 \\
\hline
\end{tabular}
\end{center}
\caption{When each bound is the best for $13 \le q \le 29$.}\label{tabStat2}
\end{table}

\begin{table}
\begin{center}
\scriptsize{
\begin{tabular}{lll|lllllll}
\hline
$q$ &$n$ &$d$&  B &  J &  H &  G &  E &  S &  L \\
\hline
 7  & 45 & 21& \bf{22} & 23 & 24 & 23 & 23 & 25 & 23 \\
 8  & 51 & 24& \bf{25} & 27 & 28 & 26 & 26 & 28 & 26 \\
 9  & 17 & 7 & \bf{10} & 11 & 11 & 11 & 11 & 11 & 11 \\
 11 & 90 & 55& \bf{30} & 41 & 42 & 32 & 35 & 36 & 31 \\
 13 & 32 & 9 & \bf{23} & 24 & 24 & 24 & 24 & 24 & 25 \\
 16 & 52 & 14& \bf{38} & 39 & 40 & 39 & 39 & 39 & 41 \\
 17 & 38 & 9 & \bf{29} & 30 & 30 & 30 & 30 & 30 & 31 \\
 19 & 42 & 9 & \bf{33} & 34 & 34 & 34 & 34 & 34 & 36 \\
 23 & 91 &17 & \bf{74} & 75 & 75 & 75 & 75 & 75 & 78 \\
 25 & 31 & 5 & \bf{26} & 27 & 27 & 27 & 27 & 27 & 28 \\
 27 & 88 &24 & \bf{64} & 66 & 67 & 65 & 66 & 65 & 69 \\
 29 &100 &29 & \bf{71} & 74 & 74 & 72 & 74 & 72 & 76 \\
\hline
\end{tabular}
}
\end{center}
\caption{Some cases where Bound $\B$ beats all the other bounds in the range $7 \le q \le 29$.}
\label{tabBnewbounds}
\end{table}

\begin{table}
\begin{center}
\begin{tabular}{p{0.18\textwidth}|llllllll}
\hline
$q$ & 2 & 3 & 4 & 5 & 7 & 8 & 9 & 11\\
\hline
Draws(D) (\%) & 38.02 & 31.20 & 31.20 & 31.94 & 40.54 & 47.59 & 53.44 & 64.80\\
Wins(W) (\%) & 0 & 0 & 0 & 0 & 0.1894 & 1.052 & 1.830 & 3.955\\
$\delta=0$ (\% over D+W) & 44.67 & 71.14 & 61.77 & 59.82 & 68.75 & 74.22 & 79.71 & 85.43\\
Use of Plotkin (\% over D+W) & 41.50 & 69.52 & 56.84 & 51.98 & 57.02 & 61.16 & 65.09 & 68.57 \\
Maximum $d/n$ in wins & - & - & - & - & 0.47 & 0.48 & 0.52 & 0.63 \\
Plotkin Range $d/n$ & 0.50 & 0.67 & 0.75 & 0.80 & 0.87 & 0.88 & 0.89 & 0.91 \\
\hline
\end{tabular}
\end{center}
\caption{Statistics for Bound $\B$ for $2 \le q \le 11$.}
\label{tabB3-100}
\end{table}

\begin{table}
\begin{center}
\begin{tabular}{p{0.18\textwidth}|llllllll}
\hline
$q$ & 13 & 16 & 17 & 19 & 23 & 25 & 27 & 29\\
\hline
Draws(D) (\%) & 73.11 & 77.41 & 78.48 & 77.84 & 73.43 & 71.18 & 69.60 & 69.60\\
Wins(W) (\%) & 3.514 & 4.208 & 5.113 & 7.574 & 14.69 & 17.55 & 19.80 & 21.19\\
$\delta=0$ (\% over D+W) & 87.96 & 88.45 & 88.49 & 88.28 & 85.00 & 83.02 & 80.89 & 78.62\\
Use of Plotkin (\% over D+W) & 65.54 & 67.00 & 66.09 & 61.80 & 55.16 & 51.98 & 48.98 & 46.54 \\

Maximum $d/n$ in wins & 0.634 & 0.640 & 0.486 & 0.487 & 0.489 & 0.490 & 0.491 & 0.492 \\
Plotkin Range $d/n$ & 0.92 & 0.94 & 0.94 & 0.95 & 0.96 & 0.96 & 0.96 & 0.97 \\
\hline
\end{tabular}
\end{center}
\caption{Statistics for Bound $\B$ for $13 \le q \le 29$.}\label{tabB3-100bis}
\end{table}

\end{document}